\documentclass[letterpaper, 10 pt, conference]{ieeeconf}%

\usepackage{graphicx}
\usepackage{epsfig}
\usepackage{times}
\usepackage{amsmath}
\usepackage{thmtools,thm-restate}
\usepackage{amssymb}
\usepackage[section]{placeins}
\usepackage{placeins}
\usepackage{amsmath}
\usepackage{multirow}
\usepackage{graphicx}
\usepackage[normalem]{ulem}
\usepackage{subcaption}
\usepackage{algorithm,tabularx}
\usepackage{algpseudocode}
\usepackage{amsfonts}
\usepackage{cases}
\usepackage{romannum}
\usepackage{textcomp}
\usepackage{xcolor}%
\usepackage{textcomp}
\providecommand{\U}[1]{\protect\rule{.1in}{.1in}}
\providecommand{\U}[1]{\protect\rule{.1in}{.1in}}

\IEEEoverridecommandlockouts
\newtheorem{assumption}{Assumption}
\newtheorem{theorem}{Theorem}

\newtheorem{lemma}{Lemma}
\newtheorem{proposition}{Proposition}
\newtheorem{remark}{Remark}
\newtheorem{definition}{Definition}

\useunder{\uline}{\ul}{}
\makeatletter
\newcommand{\multiline}[1]{  \begin{tabularx}{\dimexpr\linewidth-\ALG@thistlm}[t]{@{}X@{}}
#1
\end{tabularx}
}
\makeatother
\usepackage{cite}

\usepackage{picins}
\usepackage{centernot}
\let\labelindent\relax
\usepackage{enumitem}
\setlist[itemize]{leftmargin=*}
\usepackage{makecell}
\usepackage[normalem]{ulem}
\usepackage{cuted}
\usepackage{hyperref}
\usepackage{stfloats}
\usepackage{lscape}

\newcommand{\R}{\mathbb{R}}
\newcommand{\N}{\mathbb{N}}

\newcommand{\T}{\top}
\newcommand{\bs}{\backslash}
\newcommand{\I}{\mathbf{I}}
\newcommand{\0}{\mathbf{0}}
\newcommand{\E}{\mathcal{E}}
\newcommand{\F}{\mathcal{F}}

\newcommand{\diag}{\text{diag}}
\newcommand{\tsup}[1]{\textsuperscript{#1}}
\newcommand{\mb}[1]{\mathbf{#1}}

\newcommand{\bm}[1]{\begin{bmatrix}#1\end{bmatrix}}

\algdef{SE}[SUBALG]{Indent}{EndIndent}{}{\algorithmicend\ }%
\algtext*{Indent}
\algtext*{EndIndent}

\title{\LARGE \bf
Decentralized and Compositional Interconnection Topology Synthesis for Linear Networked Systems
}

\author{Shirantha Welikala, Hai Lin and Panos J. Antsaklis %\vspace{-6mm} 
\thanks{The support of the National Science Foundation (Grant No. CNS-1830335, IIS-2007949) is gratefully acknowledged.}
\thanks{The authors are with the Department of Electrical Engineering, College of Engineering, University of Notre Dame, IN 46556, \texttt{{\small \{wwelikal,hlin1,pantsakl\}@nd.edu}}.}}

\begin{document}

\maketitle

\pagenumbering{arabic}
\thispagestyle{plain}
\pagestyle{plain}

\begin{abstract}
In this paper, we consider networked systems comprised of interconnected sets of linear subsystems and propose a decentralized and compositional approach to stabilize or dissipativate such linear networked systems via optimally modifying some existing interconnections and/or creating entirely new interconnections. We also extend this interconnection topology synthesis approach to ensure the ability to stabilize or dissipativate such linear networked systems under distributed (local) feedback control. To the best of the authors' knowledge, this is the first work that attempts to address the optimal interconnection topology synthesis problem for linear networked systems. 
The proposed approach in this paper only involves solving a sequence of linear matrix inequality problems (one at each subsystem). Thus, using standard convex optimization toolboxes, it can be implemented efficiently and scalably in a decentralized and compositional manner.
Apart from many generic linear networked systems applications (e.g., power grid control), a unique application for the proposed interconnection topology synthesis approach is in generating random stable (or dissipative, stabilizable, dissipativate-able) linear networked systems for simulation purposes.  
We also include an interesting case study where the proposed interconnection topology synthesis approach is compared with an alternative approach that only uses dissipativity information of the involved subsystems.  
\end{abstract}

% \newpage
% \vspace{-1mm}
\section{Introduction}\label{Sec:Introduction}

% General comment about networked systems
In recent years, attention towards analysis, controller synthesis, topology synthesis as well as optimization of large-scale networked systems (comprised of dynamically coupled subsystems)  has been renewed due to their various emerging applications (e.g., in critical infrastructure networks like supply chains \cite{Ivanov2018}, power grids \cite{Tang2021}, etc.) and confronting unique challenges (e.g., resilience \cite{Tordecilla2021}, security \cite{Samad2017}, etc.).

For such networked systems, a large number of distributed control solutions have been proposed in the literature that can not only stabilize but also optimize some performance metrics of interest \cite{Antonelli2013} during their operation. However, almost all such distributed control solutions are synthesized by a centralized design process which raises concerns related to their security, scalability, and compositionality \cite{WelikalaP32022}.

Over the years, there have been several attempts to address this decentralized controller synthesis problem exploiting weak couplings \cite{Michel1983}, hierarchical techniques \cite{Ishizaki2021}, and decomposition techniques \cite{Agarwal2021}. In particular, the work in \cite{Agarwal2021} proposes a natural and efficient decomposition technique for analysis and synthesis of distributed controllers inspired by the Sylvester's criterion \cite{Antsaklis2006}. Motivated by the attractive qualities of this Sylvester's criterion based decomposition approach \cite{Agarwal2021}, our recent work in \cite{WelikalaP32022} (and its extension \cite{Welikala2022Ax2}) generalized it so that many fundamental linear control solutions (e.g., dissipativity analysis, linear observer design, etc.) can be implemented in a decentralized as well as compositional manner over large-scale linear networked systems.

Nevertheless, a major challenge faced by this approach (as well as many other control solutions proposed for large-scale networked systems) is the incompatibility between the considered networked system and the proposed solution. Such an incompatibility may be due to the inherent weaknesses in the networked system and/or in the proposed solution. For example, a networked system may not yield a conclusive (and desired) result under a particular analysis technique. Similarly, a networked system may not be capable of yielding desired properties under a particular class of controllers. 

To address this incompatibility issue, we can either change the networked system to match the proposed solution (e.g., see \cite{WelikalaP52022}), or improve/specialize the proposed solution so as to handle the considered networked system (e.g., see \cite{WelikalaP32022}). While in many scenarios it is natural and practical (and even advisable) to take the latter approach, in some instances, the prior approach is also a valid and sensible option to take. Most importantly, developing techniques to systematically change the networked systems can lead to insightful findings. For example, assuming the proposed control solution sufficiently rich, we might be able to answer questions like: What kinds of network topologies are more robust to the disturbances? What are the most critical interconnections in the networked system? What is the most cost efficient network topology?

In this paper, we set out to solve the said incompatibility issue faced by the Sylvester's criterion based decentralized and compositional approach proposed in \cite{WelikalaP32022} (intended for analysis and distributed controller synthesis of large-scale linear networked systems). To this end, we propose to change the networked system so that it matches the approach proposed in \cite{WelikalaP32022}. In particular, in the considered networked system, we treat some inter-subsystem interconnections (if not all) as design variables and explore the possibility to: (1) change those variable interconnections from their nominal values, (2) create entirely new interconnections, and/or (3) remove existing interconnections, such that the proposed approach in \cite{WelikalaP32022} can yield conclusive as well as desired results. In essence, this can be seen as an effort to synthesize the interconnection topology for linear networked systems.

In fact, there have been only very few attempts on designing interconnection topologies for networked systems. For example, the work in \cite{Rafiee2010} considers designing a network topology to make the communications optimally efficient for a continuous-time average consensus protocol. The proposed solution in \cite{Rafiee2010} takes the form of a mixed integer semidefinite program - which does not scale well. The interconnection matrix synthesis problem is considered limited to linear and positive networked systems in \cite{Ebihara2017}. Several other interconnection matrix synthesis techniques such as the ones proposed in \cite{Rufino2018,Xue2013} and \cite{Cremean2003} have been reviewed in our recent work \cite{WelikalaP52022} (see also its extension \cite{Welikala2022Ax3}). 

In particular, the work in \cite{WelikalaP52022} proposes an interconnection matrix synthesis technique for non-linear networked systems using only the subsystem dissipativity properties (i.e., without using the complete knowledge of the non-linear subsystem dynamics). However, in this paper, we limit to linear networked systems and use the complete knowledge of the linear subsystem dynamics for interconnection topology synthesis. Nevertheless, as we will show in this paper (particularly in our case study), there is a clear advantage due to the use of such additional information regarding the networked system as compared to \cite{WelikalaP52022}.

\subsubsection{\textbf{Contributions}}
Our contributions can be summarized as follows: 
(1) We take a control theoretic approach to formulate several interconnection topology synthesis problems arising related to linear networked systems as LMI problems;
(2) Since the proposed interconnection topology synthesis approach is inspired by \cite{WelikalaP32022}, it is inherently decentralized and compositional;
(3) Moreover, it can be used in scenarios where the analysis and controller synthesis approaches proposed in \cite{WelikalaP32022} returns inconclusive;
(4) We also provide candidate local objective functions that can be used to penalize deviations from a nominal set of specifications (topology).
(5) The proposed interconnection topology synthesis approach can be used to generate random linear networked systems with certain qualities (e.g., stabilizability) - which is helpful when designing, testing, and validating control strategies developed for networked systems. 
(6) Similar to \cite{WelikalaP32022}, the proposed approach can be extended to address a wide range of problems arising related to linear networked systems based on fundamental linear systems theory (e.g., optimal topology synthesis to ensure observability). 
(7) We provide candidate local objective functions that can be used to penalize deviations from a nominal interconnection topology;
(8) We provide a detailed case study comparing the interconnection topology synthesis approaches proposed in this paper and in our recent work \cite{WelikalaP52022};

\subsubsection{\textbf{Organization}} 
This paper is organized as follows. Section \ref{Sec:ProblemFormulation} presents the details of the considered class of networked systems and motivates the need for interconnection topology synthesis. Section \ref{Sec:Preliminaries} summarizes several important preliminary concepts. Our main theoretical results that address several different interconnection topology synthesis problems of interest are presented in Sec. \ref{Sec:MainResults} along with several important remarks. A case study with fundamental details, numerical results, discussions, and comparisons are provided in Sec. \ref{Sec:CaseStudy} before concluding the paper in Sec. \ref{Sec:Conclusion}.

\subsubsection{\textbf{Notation}}
The sets of real and natural numbers are denoted by $\R$ and $\N$, respectively. We define $\N_N\triangleq\{1,2,\ldots,N\}$ where $N\in\N$. 
An $n\times m$ block matrix $A$ can be represented as $A=[A_{ij}]_{i\in\N_n, j\in\N_m}$ where $A_{ij}$ is the $(i,j)$\tsup{th} block of $A$ (for indexing purposes, either subscripts or superscripts may be used, i.e., $A_{ij} \equiv A^{ij}$). 
% $[A_{ij}]_{j\in \N_m}$ and $\diag(A_{ii}:i\in\N_n)$ represent a block row matrix and a block diagonal matrix, respectively. We define $\{A_i\} \triangleq \{A_{ii}\}\cup\{A_{ij},j\in\N_{i-1}\}\cup\{A_{ji}:j\in\N_i\}$. 
If $\Psi\triangleq[\Psi^{kl}]_{k,l \in \N_m}$ where $\Psi^{kl}\triangleq[\Psi^{kl}_{ij}]_{i,j\in\N_n}$, its block element-wise form \cite{WelikalaP32022} is denoted as $\mbox{BEW}(\Psi) \triangleq [[\Psi^{kl}_{ij}]_{k,l\in\N_m}]_{i,j\in\N_n}$. 
The transpose of a matrix $A$ is denoted by $A^\T$ and $(A^\T)^{-1} = A^{-\T}$. 
The zero and identity matrices are denoted by $\0$ and $\I$, respectively (dimensions will be clear from the context). A symmetric positive definite (semi-definite) matrix $A\in\R^{n\times n}$ is represented as $A=A^\T>0$ ($A=A^\T \geq 0$). Unless stated otherwise, we assume $A>0 \iff A=A^\T>0$ (i.e., symmetry is implied by the positive definiteness). 
% The symbol $\star$ represents redundant conjugate matrices (e.g., $A^\T B\, \star \equiv A^\T B A$). 
% The symmetric part of a matrix $A$ is defined as $\H(A) \triangleq A+A^\T$ and $\H(A_{ij}) \triangleq A_{ij}+A_{ji}^\T$. 
$\mb{1}_{\{\cdot\}}$ is the indicator function and $e_{ij} \triangleq \I \cdot \mb{1}_{\{i=j\}}$.

% \newpage
\section{Problem Formulation}\label{Sec:ProblemFormulation}

\subsection{The Networked System}
We consider a networked dynamical system $\mathcal{G}_N$ comprised of $N$ interconnected subsystems denoted by $\{\Sigma_i:i\in\N_N\}$. The dynamics of the $i$\tsup{th} subsystem $\Sigma_i, i\in\N_N$ are given by 
\begin{equation}\label{Eq:SubsystemDynamics}
\begin{aligned}
    \dot{x}_i(t) =& \sum_{j\in\bar{\E}_i}A_{ij}x_j(t) + \sum_{j\in\bar{\E}_i}B_{ij}u_j(t)+\sum_{j\in\bar{\E}_i}E_{ij}w_{j}(t),\\ 
    y_i(t) =& \sum_{j\in\bar{\E}_i}C_{ij}x_j(t) + \sum_{j\in\bar{\E}_i}D_{ij}u_j(t) + \sum_{j\in\bar{\E}_i}F_{ij}w_j(t),
\end{aligned}
\end{equation}
where $x_i(t) \in \R^{n_i},\ u_i(t)\in\R^{p_i},\ w_i(t)\in\R^{q_i}$ and $y_i(t)\in\R^{m_i}$ respectively represents the state, input, disturbance and output specific to the subsystems $\Sigma_i$ at time $t\in\R_{\geq 0}$. 
In \eqref{Eq:SubsystemDynamics}, $\bar{\E}_i \triangleq \E_i\cup\{i\}$ where $\E_i \subset \N_N$ is the set of ``in-neighbors'' of the subsystem $\Sigma_i$. Formally, any subsystem $\Sigma_j$ is an ``in-neighbor'' of the subsystem $\Sigma_i$ (i.e., $j\in\E_i$) iff the matrices $A_{ij},B_{ij},C_{ij},D_{ij},E_{ij},F_{ij}$ in \eqref{Eq:SubsystemDynamics} are not all zero matrices. Conversely, $\bar{\F}_i \triangleq \F_i\cup\{i\}$ where $\F_i \triangleq \{j: j\in\N_N, \E_j \ni i \}$ is the set of ``out-neighbors'' of the subsystem $\Sigma_i$. An example networked system can be seen in Fig. \ref{Fig:LinearNetworkedSystem}. 

\begin{figure}[!b]
    \centering
    \includegraphics[width=1.75in]{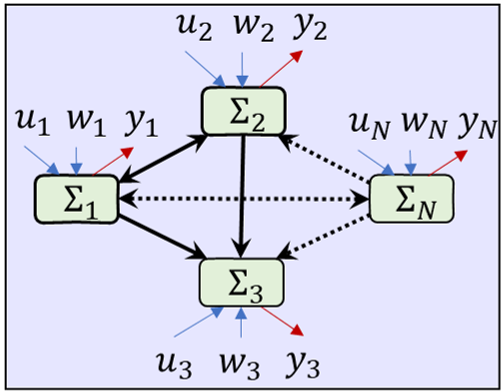}
    \caption{An example networked dynamical system $\mathcal{G}_N$.}
    \label{Fig:LinearNetworkedSystem}
\end{figure}

By writing \eqref{Eq:SubsystemDynamics} for all $i\in\N_N$ and concatenating suitably, we can get the dynamics of the networked system $\mathcal{G}_N$ as 
\begin{equation}\label{Eq:NetworkedSystemDynamics}
\begin{aligned}
\dot{x}(t) =&\ Ax(t) + Bu(t) + Ew(t),\\
y(t) =&\ Cx(t) + Du(t) + Fw(t),
\end{aligned}
\end{equation}
where $A=[A_{ij}]_{i,j\in\N_N}$, $B=[B_{ij}]_{i,j\in\N_N}$, $E=[E_{ij}]_{i,j\in\N_N}$, $C=[C_{ij}]_{i,j\in\N_N}$, $D=[D_{ij}]_{i,j\in\N_N}$ and $F=[F_{ij}]_{i,j\in\N_N}$ are all $N \times N$ block matrices, and $x(t)\in \R^n,\ u(t)\in\R^p, w(t)\in\R^q$ and $y(t)\in\R^m$ (with $n = \sum_{i\in\N_N} n_i$, $p = \sum_{i\in\N_N}p_i$, $q = \sum_{i\in\N_N} q_i$ and $m = \sum_{i\in\N_N} m_i$) are all $N\times 1$ block matrices respectively representing the networked system's state, input, disturbance and output at time $t\in\R_{\geq 0}$.  

\subsection{Distributed Controllers}
To enforce desired properties (e.g., stability) upon the networked system, a subsystem $\Sigma_i, i\in\N_N$ can use a distributed state feedback controller:
\begin{equation}\label{Eq:LocalStateFeedbackControl}
    u_i(t) = \sum_{j\in\bar{\E}_i} K_{ij}x_j(t).
\end{equation}

By writing \eqref{Eq:LocalStateFeedbackControl} for all $i\in\N_N$ and concatenating appropriately, we get the global form of the distributed feedback controller as 
\begin{equation}
    \label{Eq:GlobalControl}
    u(t) = Kx(t),   
\end{equation}
where $K=[K_{ij}]_{i,j\in\N_N}$. 

Note that the unspecified blocks in various block matrices in both \eqref{Eq:NetworkedSystemDynamics} and \eqref{Eq:GlobalControl} are zeros matrices (e.g., $A_{ij}=0, \forall j\not\in\E_i$).  

\subsection{Interconnection Topology Synthesis}

Even though state feedback control is a reasonable approach to enforce desired properties (e.g., stability) upon the networked system, it may be not useful in two scenarios: (1) when the networked system inherently involves no control inputs (i.e., when $B=D=\0$ in \eqref{Eq:NetworkedSystemDynamics}), or (2) when the networked system is inherently incapable of achieving the desired properties under state feedback control (e.g., if \eqref{Eq:NetworkedSystemDynamics} is not stabilizable when the desired property is stability). To address these inherent weaknesses of the networked system, in this paper, we propose to optimally adjust the interconnection parameters of the networked system (mainly $A_{ij}$ blocks with $i \neq j$ in \eqref{Eq:SubsystemDynamics}). Hence this approach can be seen as an attempt to synthesize the interconnection topology of the networked system.

Note also that, for the purposes of analysis and controller synthesis of the networked system \eqref{Eq:NetworkedSystemDynamics}, we can use the decentralized and compositional technique proposed in \cite{WelikalaP32022}. However, due to the used assumptions in \cite{WelikalaP32022}, this decentralized and compositional technique can return inconclusive (when analyzing) or infeasible (when synthesizing controllers) \cite{Welikala2022Ax2}. Nevertheless, as we will show in the sequel, this technical weakness can also be addressed by optimally adjusting the interconnection parameters of the networked system.

\section{Preliminaries} \label{Sec:Preliminaries}

\subsection{Stability and Dissipativity}

Since our main goal is to synthesize the interconnection topology of the linear networked system \eqref{Eq:NetworkedSystemDynamics} so as to enforce properties like stability or dissipativity (both without or with distributed feedback control \eqref{Eq:LocalStateFeedbackControl}), we next briefly introduce some relevant stability and dissipativity results.

%% System dynamics
Consider the linear time-invariant (LTI) system  
\begin{equation}\label{Eq:LTISystem}
\begin{aligned}
    \dot{x}(t) = Ax(t) + Bu(t),\\
    y(t) = Cx(t) + Du(t),
\end{aligned}
\end{equation}
where $x(t)\in\R^n,u(t)\in\R^p$, and $y(t)\in\R^m$ respectively represent the state, control input, and output at time $t\in\R_{\geq0}$.

%% Stability
\paragraph{\textbf{Stability}} A well-known necessary and sufficient condition for the stability of \eqref{Eq:LTISystem} is given in the following lemma as a linear matrix inequality (LMI).

\begin{lemma}\cite{Antsaklis2006} \label{Lm:Stability}
The dynamical system \eqref{Eq:LTISystem} (under $u(t)=\0$) is globally uniformly (exponentially) stable iff $\exists P > 0$ such that 
\begin{equation}\label{Eq:Lm:Stability}
    -A^\T P - PA \geq 0\ \ \ \ \ \ \mbox{($-A^\T P - PA > 0$)}.
\end{equation}
\end{lemma}

Note that, henceforth, by `stability,' we simply refer to global exponential stability.

%% Dissipativity
\paragraph{\textbf{$(Q,S,R)$-Dissipativity}} In general, dissipativity is an important property of dynamical systems that has many practical uses \cite{Willems1972a}. In this paper, we consider the quadratic dissipativity property called $(Q,S,R)$-dissipativity. 

\begin{definition} \cite{Kottenstette2014} \label{Def:QSRDissipativity}
The dynamical system \eqref{Eq:LTISystem} is $(Q,S,R)$-dissipative from $u(t)$ to $y(t)$, if there exists a positive definite function $V(x):\R^n\rightarrow\R_{\geq0}$ (storage function) such that for all $t_1 \geq t_0 \geq 0, x(t_0) \in \R^n$ and $u(t)\in\R^m$, the inequality
% \begin{equation*}
$
    V(x(t_1))-V(x(t_0)) \leq \int_{t_0}^{t_1} 
    \begin{bmatrix}
    y(t)\\u(t)
    \end{bmatrix}^\T
    \begin{bmatrix}
    Q & S \\ S^\T & R
    \end{bmatrix}
    \begin{bmatrix}
    y(t)\\u(t)
    \end{bmatrix}dt
$
% \end{equation*}
holds for the given $Q\in\R^{m \times m}, S\in \R^{m \times p}$ and $R\in\R^{p\times p}$.
\end{definition}

Through appropriate choices of $Q, S$ and $R$ matrices, $(Q,S,R)$-dissipativity can capture several dynamical properties of interest, as summarized in the following remark.

\begin{remark} \cite{Kottenstette2014} \label{Rm:QSRDissipativityVariants}
The dynamical system \eqref{Eq:LTISystem} satisfying Def. \ref{Def:QSRDissipativity}:
(i) is \emph{passive} iff $Q=0, S=\frac{1}{2}\I, R=0$;
(ii) is \emph{strictly passive} iff $Q=-\rho \I, S=\frac{1}{2}\I, R=-\nu \I$ where $\rho,\nu>0$ ($\nu$, $\rho$ are passivity indices \cite{Welikala2022Ax2});
(iii) is \emph{$\mathcal{L}_2$-stable} iff $Q=-\I, S=0, R=-\gamma^2\I$ where $\gamma \geq 0$ ($\gamma$ is an \emph{$\mathcal{L}_2$-gain} of the system);
(iv) is \emph{sector bounded} iff $Q=-\I, S=(a+b)\I, R=-ab\I$ where $a,b\in\R$ ($a,b$ are sector bound parameters).
\end{remark}

A necessary and sufficient condition for $(Q,S,R)$-dissipativity of \eqref{Eq:LTISystem} is given in the next lemma as an LMI.
 
\begin{lemma}\cite{Welikala2022Ax2}
\label{Lm:QSRDissipativity}
The dynamical system \eqref{Eq:LTISystem} is $(Q,S,R)$-dissipative ($-Q>0,R=R^\T$) from $u(t)$ to $y(t)$ iff $\exists P>0$ such that 
\begin{equation}\label{Eq:Lm:QSRDissipativity}
    \begin{bmatrix}
    -A^\T P - P A     &  -PB + C^\T S            & C^\T \\
    -B^\T P + S^\T C          & D^\T S + S^\T D + R    & D^\T \\ 
    C & D & -Q^{-1}
    \end{bmatrix} \geq 0.
\end{equation}
\end{lemma}

Note that LMIs in \eqref{Eq:Lm:Stability} and \eqref{Eq:Lm:QSRDissipativity} are ``linear'' as they contain linear terms in the corresponding design variable $P$. As shown in \cite{Boyd1994}, LMIs can be solved efficiently and scalably using standard convex optimization algorithms.  

\subsection{Interconnection Topology Synthesis}
Due to the similarity between \eqref{Eq:LTISystem} and \eqref{Eq:NetworkedSystemDynamics}, LMIs \eqref{Eq:Lm:Stability} and \eqref{Eq:Lm:QSRDissipativity} can respectively be used for stability and dissipativity analysis of the networked system \eqref{Eq:NetworkedSystemDynamics}. Note, however, that, in the LMIs \eqref{Eq:Lm:Stability} and \eqref{Eq:Lm:QSRDissipativity}, we cannot treat the matrix $A$ (particularly its non-diagonal elements $A_{ij}$ with $i\neq j$) as an independent design variable separately from $P$. This is because it makes  \eqref{Eq:Lm:Stability} and \eqref{Eq:Lm:QSRDissipativity} bi-linear matrix inequalities - which are non-linear and significantly harder to solve compared to the corresponding LMIs. Therefore, synthesizing certain elements of $A$ (i.e., the interconnection parameters of the networked system) such that stability or dissipativity holds for the networked system \eqref{Eq:NetworkedSystemDynamics} is a non-trivial and challenging problem. Similarly, synthesizing certain interconnection parameters of the networked system such that stabilizability or dissipativate-ability holds for the networked system \eqref{Eq:NetworkedSystemDynamics} under state feedback control \eqref{Eq:LocalStateFeedbackControl} is also a non-trivial and challenging problem.

We address these challenges by taking a decentralized and compositional approach. In particular, to analyze or enforce (via state feedback control) desired properties like stability or dissipativity of the networked system \eqref{Eq:NetworkedSystemDynamics}, compared to solving large centralized LMIs like \eqref{Eq:Lm:Stability} and \eqref{Eq:Lm:QSRDissipativity}, we propose to solve their small decentralized and compositional versions proposed in \cite{WelikalaP32022}. This approach allows us to sequentially synthesize the interconnection parameters of the networked system \eqref{Eq:NetworkedSystemDynamics} (i.e., step-by-step). In particular, at each step, we add a new subsystem to the current network and solve a small LMI problem where some interconnection parameters related to the new subsystem are treated as design variables while all other interconnection parameters are treated as fixed.  

Before providing more details about this approach, we first need to outline the decentralized and compositional approach proposed in \cite{WelikalaP32022} that can be used to analyze/enforce centralized LMIs exploiting a concept named ``network matrices.''

% Intuitively, verifying/enforcing such LMI conditions require the knowledge of the complete networked system \eqref{Eq:NetworkedSystemDynamics} and thus calls for a centralized entity. Moreover, if we introduce new subsystems into the network, the complete verification/enforcement process has to be re-evaluated. To address these challenges, we make the objective of this paper to design a systematic distributed and compositional approach to verify/enforce different LMI conditions of interest (not limited to \eqref{Eq:Lm:Stability} and \eqref{Eq:Lm:QSRDissipativity}), for the networked system \eqref{Eq:NetworkedSystemDynamics}.

% A critical feature that will be exploited here is that LMI conditions of interest now involves matrices structured according to the network topology. The next section of this paper focuses on such network-related matrices (that we refer to as ``network matrices'') and derives a distributed and compositional test criterion to evaluate their \emph{positive definiteness}. Subsequently, we will see that, such a test criterion can easily be adapted to verify/enforce LMI conditions for networked systems in a distributed and compositional manner.  
\subsection{Network Matrices}
Consider a directed network $\mathcal{G}_n=(\mathcal{V},\mathcal{E})$ where $\mathcal{V} \triangleq \{\Sigma_i:i\in\N_n\}$ is the set of subsystems (nodes), $\mathcal{E} \subset \mathcal{V}\times \mathcal{V}$ is the set of inter-subsystem interconnections (edges) and $n\in\N$. We next recall a class of matrices named ``network matrices'' introduced in \cite{WelikalaP32022} corresponding to such a network $\mathcal{G}_n$.

\begin{definition}\cite{WelikalaP32022}\label{Def:NetworkMatrices}
	Given a network $\mathcal{G}_n=(\mathcal{V},\mathcal{E})$, any $n\times n$ block matrix $\Theta = \bm{\Theta_{ij}}_{i,j\in\N_n}$ is a corresponding \emph{network matrix} if: 
	(1) $\Theta_{ij}$ contains information specific only to the subsystems $\Sigma_i$ and $\Sigma_j$, and 
	(2) $(\Sigma_i,\Sigma_j) \not\in \mathcal{E}$ and $(\Sigma_j,\Sigma_i)\not\in \mathcal{E}$ implies $\Theta_{ij}=\Theta_{ji}=\0$, for all $i,j\in\N_n$.
\end{definition}

% Key properties of network matrices:
% Block element-wise form, and enforcing positive definiteness
% Things from the MED paper

According to this definition, any $n \times n$ block matrix $\Theta=[\Theta_{ij}]_{i,j\in\N_n}$ is a network matrix of $\mathcal{G}_n$ if $\Theta_{ij}$ is a coupling weight matrix corresponding to the edge $(\Sigma_i,\Sigma_j)\in\mathcal{V}$. Moreover, any $n \times n$ block diagonal matrix $\Theta=\diag(\Theta_{ii}:i\in \N_n)$ where $\Theta_{ii}$ contains information specific only to the subsystem $\Sigma_i$, is a network matrix of any network with $n\in\N$ subsystems. The following lemmas provide several useful properties of such network matrices established in \cite{WelikalaP32022}.

\begin{lemma}
	\label{Lm:NetworkMatrixProperties}
	\cite{WelikalaP32022}
	Given a network $\mathcal{G}_n$, a few corresponding block network matrices $\Theta,\Phi,\{\Psi^{kl}:k,l\in\N_m\}$, and some arbitrary block-block matrix $\Psi\triangleq[\Psi^{kl}]_{k,l \in \N_m}$:
	\begin{enumerate}
		\item $\Theta^\T$, \ $\alpha \Theta + \beta \Phi$ are network matrices for any $\alpha,\beta \in \R$.
		\item $\Phi \Theta$, $\Theta\Phi$ are network matrices whenever $\Phi$ is a block diagonal network matrix.
		\item $\mbox{BEW}(\Psi)\triangleq [[\Psi^{kl}_{ij}]_{k,l\in\N_m}]_{i,j\in\N_n}$ is a network matrix.
	\end{enumerate}
\end{lemma}

The above lemma enables claiming custom block matrices as ``network matrices'' by enforcing additional conditions. For example, if $A$ and $P$ are two block network matrices and $P$ is block diagonal, then: (1) $A^\T P, P A$ and $A^\T P + PA$ are all network matrices, and (2) if $\scriptsize \Psi\triangleq\bm{P & A^\T P\\ PA & P}$ is some block-block matrix, its \emph{block element-wise} (BEW) form $\mbox{BEW}(\Psi)\triangleq$ $\scriptsize \bm{\bm{P_{ii}e_{ij} & M_{ji}^\T P_{jj} \\ P_{ii}M_{ij} & P_{ii}e_{ij}}}_{i,j\in\N_N}$ is a network matrix. 

\begin{lemma}\label{Lm:AlternativeLMI_BEW}
	\cite{WelikalaP32022}
	Let $\Psi = [\Psi^{kl}]_{k,l\in\N_m}$ be an $m\times m$ block-block matrix where $\Psi^{kl}, \forall  k,l\in\N_m$ are $n \times n$ block matrices. Then, $\Psi > 0 \iff \text{BEW}(\Psi) \triangleq [[\Psi^{kl}_{ij}]_{k,l\in\N_m}]_{i,j\in\N_n}>0$.
\end{lemma}

Inspired by Sylvester's criterion \cite{Antsaklis2006}, the following lemma provides a decentralized and compositional testing criterion to evaluate the positive definiteness of an $N\times N$ block matrix $W=[W_{ij}]_{i,j\in\N_N}$ (for more details, see  \cite{Welikala2022Ax2}).  

\begin{lemma}\label{Lm:MainLemma}
	\cite{WelikalaP32022}
	A symmetric $N \times N$ block matrix $W = [W_{ij}]_{i,j\in\N_N} > 0$ if and only if
	\begin{equation}\label{Eq:Lm:MainLemma1}
		\tilde{W}_{ii} \triangleq W_{ii} - \tilde{W}_i \mathcal{D}_i \tilde{W}_i^\T > 0,\ \ \ \ \forall i\in\N_N,
	\end{equation}
	where
	\begin{equation}\label{Eq:Lm:MainLemma2}
		\begin{aligned}
			\tilde{W}_i \triangleq&\ [\tilde{W}_{ij}]_{j\in\N_{i-1}} \triangleq W_i(\mathcal{D}_i\mathcal{A}_i^\T)^{-1},\\
			W_i \triangleq&\  [W_{ij}]_{j\in\N_{i-1}}, \ \ \ 
			\mathcal{D}_i \triangleq \diag(\tilde{W}_{jj}^{-1}:j\in\N_{i-1}),\\
			\mathcal{A}_i \triangleq&\ 
			\bm{
				\tilde{W}_{11} & \0 & \cdots & \0 \\
				\tilde{W}_{21} & \tilde{W}_{22} & \cdots & \0\\
				\vdots & \vdots & \vdots & \vdots \\
				\tilde{W}_{i-1,1} & \tilde{W}_{i-1,2} & \cdots & \tilde{W}_{i-1,i-1}
			}.
		\end{aligned}
	\end{equation}
\end{lemma}

The above lemma shows that testing positive definiteness of an $N\times N$ block matrix $W=[W_{ij}]_{i,j\in\N_N}$ can be broken down to $N$ separate smaller sequence of tests (iterations). In a network setting where $W$ is a block network matrix corresponding to a network $\mathcal{G}_N$, at the $i$\tsup{th} iteration (i.e., at the subsystem $\Sigma_i$), we now only need to test whether $\tilde{W}_{ii}>0$, where $\tilde{W}_{ii}$ can be computed using: 
(1) $\{W_{ij}:j \in \N_i\}$ (related blocks to the subsystem $\Sigma_i$ extracted from $W$), 
(2) $\{\tilde{W}_{ij}:j\in \N_{i-1}\}$ (computed using \eqref{Eq:Lm:MainLemma2} at subsystem $\Sigma_i$), and 
(3) $\{\{\tilde{W}_{jk}:k\in\N_j\}:j\in\N_{i-1}\}$ (blocks computed using \eqref{Eq:Lm:MainLemma2} at previous iteratins/subsystems $\{\Sigma_j:j\in\N_{i-1}\}$). 
Note also that, using Schur complement theory, the matrix inequality $\tilde{W}_{ii}>0$ \eqref{Eq:Lm:MainLemma1} can be transformed into a form that is linear in $[W_{ij}]_{j\in\N_{i-1}}$.

Therefore it is clear that Lm. \ref{Lm:MainLemma} can be used to efficiently test/enforce the positive definiteness of a network matrix in a decentralized manner. In fact, as shown in \cite{Welikala2022Ax2}, for some network topologies, this process is fully distributed (i.e., no communications are required between non-neighboring subsystems). Moreover, the compositionally of this process (i.e., the resilience to subsystem removals/additions from/to the network) has also been established in \cite{WelikalaP32022}. This decentralized and compositional approach to test/enforce the positive-definiteness of a network matrix $W$ is outlined in Alg. \ref{Alg:DistributedPositiveDefiniteness}.

\section{Main Results}\label{Sec:MainResults}

\begin{algorithm}[!t]
	\caption{Testing/Enforcing $W>0$ in a Network Setting.}
	\label{Alg:DistributedPositiveDefiniteness}
	\begin{algorithmic}[1]
		\State \textbf{Input: } $W = [W_{ij}]_{i,j\in\N_N}$
		\State \textbf{At each subsystem $\Sigma_i, i \in \N_N$ execute:} 
		\Indent
		\If{$i=1$}
		\State Test/Enforce: $W_{11}>0$
		\State Store: $\tilde{W}_1 \triangleq [W_{11}]$ \Comment{To be sent to others.}
		\Else
		\State \textbf{From each subsystem $\Sigma_j, j\in\N_{i-1}$:}
		\Indent
		\State Receive: $\tilde{W}_j \triangleq  [\tilde{W}_{j1},\tilde{W}_{j2},\ldots,\tilde{W}_{jj}]$
		\State Receive: Required info. to compute $W_{ij}$
		\EndIndent
		\State \textbf{End receiving}
		\State Construct: $\mathcal{A}_i, \mathcal{D}_i$ and $W_i$.  
		\Comment{Using: \eqref{Eq:Lm:MainLemma2}.}
		\State Compute: $\tilde{W}_i \triangleq W_i (\mathcal{D}_i\mathcal{A}_i^\T)^{-1}$ 
		\Comment{From \eqref{Eq:Lm:MainLemma2}.}
		\State Compute: $\tilde{W}_{ii} \triangleq W_{ii} - \tilde{W}_{i}\mathcal{D}_i\tilde{W}_i^\T$ 
		\Comment{From \eqref{Eq:Lm:MainLemma1}.}
		\State Test/Enforce: $\tilde{W}_{ii} > 0$
		\State Store: $\tilde{W}_i \triangleq [\tilde{W}_i, \tilde{W}_{ii}]$
		\Comment{To be sent to others}
		\EndIf
		\EndIndent
		\State \textbf{End execution}
	\end{algorithmic}
\end{algorithm}

In this section, we present our main theoretical results on decentralized and compositional synthesis of interconnection topology in linear networked systems. Note that this synthesis process is designed to enforce stability or dissipativity both without or with the help of distributed state feedback control, i.e., we are interested in enforcing: 
(1) stability,
(2) stabilizability (under feedback control),
(3) dissipativity and 
(4) dissipativate-ability (under feedback control),
via synthesizing an appropriate interconnection topology. 

Based on the subsystem dynamics \eqref{Eq:SubsystemDynamics}, it is clear that the interconnection topology of the networked system \eqref{Eq:NetworkedSystemDynamics} is determined by the block structures of the block network matrices $A, B, E, C, D, F$ and $K$ in \eqref{Eq:NetworkedSystemDynamics} and \eqref{Eq:GlobalControl} (see also Def. \ref{Def:NetworkMatrices}). Regarding these block network matrices, we make the following two technical assumptions. 

\begin{assumption}\cite{WelikalaP32022}\label{As:BlockDiagonal}
    The block network matrices $C, D$, and $F$ (in \eqref{Eq:NetworkedSystemDynamics}) are block diagonal network matrices. 
\end{assumption}

% Before providing the four theorems corresponding to these four cases, we make the following technical assumption. 

\begin{assumption}\label{As:FixedBlocks}
Any block network matrix $M = [M_{ij}]_{i,j\in\N_N}$ in the set $\{A,B,E,K\}$ (from \eqref{Eq:NetworkedSystemDynamics} and \eqref{Eq:GlobalControl}) unless stated otherwise, satisfies the following statement: Corresponding to a subsystem $\Sigma_i, i\in\N_N$, the intrinsic matrix block $M_{ii}$ and the interconnection matrix blocks $\{M_{ij}:j\in\E_i\}$ and $\{M_{ji}:j\in\F_i\}$ are known and fixed while the remaining interconnection matrix blocks $\{M_{ij}:j\not\in\E_i\}$ and $\{M_{ji}:j\not\in \F_i\}$ are free to be designed. 
\end{assumption}

Note that the above assumption relaxes a hard constraint used in \cite{WelikalaP32022}. For example, in \cite{WelikalaP32022}, $A_{ij}=\0, \forall j\not\in\E_i$ and $A_{ji}=\0, \forall j\not\in\F_i$ was assumed. In contrast, here we allow new interconnections when necessary via treating, e.g., $\{A_{ij}:j\not\in\E_i\}$ and $\{A_{ji}:j\not\in \F_i\}$ as free to be designed. 

Note also that we execute this design/synthesis task in a decentralized and compositional manner. Therefore, in its $i$\tsup{th} step (executed at the subsystem $\Sigma_i, i\in\N_N$, according to Alg. \ref{Alg:DistributedPositiveDefiniteness}), we only need to synthesize a subset of interconnection matrix blocks, E.g., $\{A_{ij}:j\not\in\E_i\cap\N_{i-1}\}$ and $\{A_{ji}:j\not\in\F_i\cap\N_{i-1}\}$. 

\subsection{Enforcing Stability and Stabilizability}

\paragraph{\textbf{Stability}} The following theorem considers an un-actuated networked system and provides how new interconnections can be created via designing the interconnection matrix blocks $\{A_{ij}:j\not\in\E_i\}$ and $\{A_{ji}:j\not\in\F_i\}$.   
 
\begin{theorem} \label{Th:CTLTIStability}(Stability via Topology Synthesis)
The networked system \eqref{Eq:NetworkedSystemDynamics}, under Assumption \ref{As:FixedBlocks}, $u(t)=\0$ and $w(t)=\0$, is stable if at each subsystem $\Sigma_i,i\in\N_N$, the LMI problem  
\begin{equation}\label{Eq:Th:CTLTIStability}
\begin{aligned}
    \mathbb{P}_1: \text{Find}& \ \ P_{ii}, \{Q_{ij}:j\not\in\E_i\cap\N_{i-1}\}, \{A_{ji}:j\not\in\F_i\cap\N_{i-1}\}\\ 
    \text{such that}& \ \ P_{ii} > 0, \ \ \tilde{W}_{ii}>0,
\end{aligned}
\end{equation}
is feasible, where $\tilde{W}_{ii}$ is computed from Alg. \ref{Alg:DistributedPositiveDefiniteness} (Steps: 3-16) when analyzing $W = [W_{ij}]_{i,j\in\N_N}>0$ with  
\begin{equation}\label{Eq:Th:CTLTIStability2}
    W_{ij} = - P_{ii}A_{ij}\mb{1}_{\{j\in\bar{\E}_i\}} -A_{ji}^\T P_{jj} \mb{1}_{\{j\in\bar{\F}_i\}} - Q_{ij}\mb{1}_{\{j\not\in\E_i\}},
\end{equation}
and the new interconnections are  
$\{A_{ji}:j\not\in\F_i\cap\N_{i-1}\}$ and 
$\{A_{ij} \triangleq P_{ii}^{-1} Q_{ij}: j\not\in\E_i\cap \N_{i-1}\}$.
\end{theorem}

\begin{proof}
Let us define $P \triangleq \diag(P_{ii}:i\in\N_N)$ and $W \triangleq -A^\T P - P A$ where now $A=[A_{ij}]_{i,j\in\N_N}$ includes variable interconnection blocks $\{A_{ij}:j\not\in\E_i\}$ and $\{A_{ji}:j\not\in\F_i\}$ in its every $i$\tsup{th} column and row, respectively, $i\in\N_N$ (replacing the fixed $\0$ blocks that were there as per the original definition of $A$ given in \eqref{Eq:NetworkedSystemDynamics}).

According to Lm. \ref{Lm:Stability}, this new networked system is stable if we can find $P>0$ such that $W>0$. Based on the above definition of $W=[W_{ij}]_{i,j\in\N_N}$, it is a symmetric network matrix (see Def. \ref{Def:NetworkMatrices}) where  
\begin{equation}\label{Eq:Th:CTLTIStabilityStep1}
    W_{ij} = - P_{ii}A_{ij} - A_{ji}^\T P_{jj}. 
\end{equation}
Thus, we can use Alg. \ref{Alg:DistributedPositiveDefiniteness} to test $W>0$ in a decentralized and compositional manner via testing $\tilde{W}_{ii}>0$ at each subsystem $\Sigma_i,i\in\N_N$ sequentially (see Lm. \ref{Lm:MainLemma} and \eqref{Eq:Lm:MainLemma1}). 

However, according to \eqref{Eq:Lm:MainLemma2}, testing $\tilde{W}_{ii}>0$ will be an LMI problem only if the terms $\{W_{ij}:j\in\N_{i-1}\}$ are linear in the program variables: $P_{ii}, \{A_{ij}:j\not\in\E_i\cap\N_{i-1}\}$ and $\{A_{ji}:j\not\in\F_i\cap\N_{i-1}\}$. Based on \eqref{Eq:Th:CTLTIStabilityStep1}, the term $-P_{ii}A_{ij}$ in $W_{ij}$ become bi-linear whenever $j\in\not\in\E_i\cap\N_{i-1}$. This calls for a change of variables:
\begin{equation}
    Q_{ij} \triangleq P_{ii}A_{ij},\ \ j\not\in \E_i\cap \N_{i-1},
\end{equation}
which transforms $W_{ij}$ in \eqref{Eq:Th:CTLTIStabilityStep1} into the form \eqref{Eq:Th:CTLTIStability2} - which is linear in terms of the new program variables: $P_{ii}, \{Q_{ij}:j\not\in\E_i\cap\N_{i-1}\}$ and $\{A_{ji}:j\not\in\F_i\cap\N_{i-1}\}$. Consequently, testing $\tilde{W}_{ii}>0$ takes the form of an LMI problem \eqref{Eq:Th:CTLTIStability}.

If all the local LMI problems \eqref{Eq:Th:CTLTIStability} are feasible, $\exists P_{ii}>0$ such that $\tilde{W}_{ii}>0, \forall i\in\N_N$ - which implies that $\exists P > 0$ such that $W>0$. This leads to the conclusion that the new networked system with new interconnections $\{A_{ji}:j\not\in\F_i\cap\N_{i-1}\}$ and 
$\{A_{ij} \triangleq P_{ii}^{-1} Q_{ij}: j\not\in\E_i\cap \N_{i-1}\}$ is stable. 
\end{proof}

\begin{remark}(Objective Functions)\label{Rm:LMIObjective}
    As the objective function of the decentralized LMI problem \eqref{Eq:Th:CTLTIStability} proposed in Th. \ref{Th:CTLTIStability}, we can use:
    \begin{equation}\label{Eq:LMIObjective}
        J = \sum_{j\not\in\E_i\cap\N_{i-1}} c_{ij} \Vert Q_{ij}-P_{ii}\bar{A}_{ij} \Vert + \beta_{ii}\sum_{j\not\in\F_i\cap\N_{i-1}} c_{ji} \Vert A_{ji} - \bar{A}_{ji}\Vert, 
    \end{equation}
    where (1) $\bar{A}_{ij}$ and $\bar{A}_{ji}$ matrices represent desired/prescribed candidates for $A_{ij}$ and $A_{ji}$, respectively,
    (2) $c_{ij}$ and $c_{ji}$ scalars represent the cost of the interconnections $(\Sigma_j,\Sigma_i)$ and $(\Sigma_i,\Sigma_j)$, respectively, and
    (3) $\beta_{ii}$ is a normalizing constant. It is easy to see that selecting $\beta_{ii} = \Vert P_{ii} \Vert$ equally weights the two components in the above objective function \eqref{Eq:LMIObjective}. However, such a choice is not practical as it makes the objective function non-convex. Therefore, a reasonable alternative would be to use $\beta_{ii} = \Vert P_{jj} \Vert$ - which is a known constant when evaluating \eqref{Eq:Th:CTLTIStability}. Note that an intuitive objective function of this form can also be used with the decentralized LMI problems proposed in the sequel in Theorems \ref{Th:StabilizationUnderFSF}-\ref{Th:DissipativationUnderFSF} (with a few minor modifications). 
\end{remark}

% \begin{remark}(Constraints on Interconnection Parameters)
%     An open problem
% \end{remark}

\begin{remark}(Refining Existing Interconnections)\label{Rm:Refining}
The proposed interconnection topology synthesis approach can also be used to refine existing interconnections. To show this, assume we are interested in refining the interconnection $A_{ij}$ for some $j\in\E_i$. First, we need to modify the set of in-neighbors of the subsystem $\Sigma_i$ such that $\E_i \rightarrow \E_i \bs \{j\}$. Next, the current value of $A_{ij}$ should be stored as the prescribed value $\bar{A}_{ij}$ in the LMI objective \eqref{Eq:LMIObjective} and then consider $A_{ij}$ as a design variable to be synthesized. Moreover, if we are interested in removing the interconnection $A_{ij}$ entirely, we need to set the cost coefficient $c_{ij}$ in the LMI objective \eqref{Eq:LMIObjective} arbitrarily high and set $\bar{A}_{ij}=\0$. Finally, via solving the LMI problem \eqref{Eq:Th:CTLTIStability} we can obtain the refined interconnection topology (with a refined $A_{ij}$ value).
\end{remark}

\paragraph{\textbf{Stabilizability}}
The next theoretical result is on enforcing stabilizability under distributed state feedback control via interconnection topology synthesis.

\begin{theorem}\label{Th:StabilizationUnderFSF}
(Stabilizability via Topology Synthesis) 
The networked system \eqref{Eq:NetworkedSystemDynamics} (where $B$ is block diagonal), under Assumption \ref{As:FixedBlocks}, $w(t)=\0$ and local state feedback control \eqref{Eq:LocalStateFeedbackControl}, is stable if at each subsystem $\Sigma_i,i\in\N_N$, the LMI problem    
\begin{equation}\label{Eq:Th:StabilizationUnderFSF}
\begin{aligned}
    \mathbb{P}_2: \text{Find}&\ \ M_{ii}, L_{ii}, \{L_{ij}:j\in \N_{i-1}\}, \{L_{ji}:j\in\N_{i-1}\}, \\
    &\ \ \{A_{ij}:j\not\in\E_i\cap\N_{i-1}\},\{Q_{ji}:j\not\in\F_i\cap\N_{i-1}\}\\ 
    \text{such that} &\ \ M_{ii} > 0, \ \ \tilde{W}_{ii}>0,
\end{aligned} 
\end{equation}
is feasible, where $\tilde{W}_{ii}$ is computed from Alg. \ref{Alg:DistributedPositiveDefiniteness} (Steps: 3-16) when enforcing $W = [W_{ij}]_{i,j\in\N_N}>0$ with 
\begin{equation}\label{Eq:Th:StabilizationUnderFSF2}
\begin{aligned}
    W_{ij} =&\ - A_{ij}M_{jj} - M_{ii}A_{ji}^\T\mb{1}_{\{j\in\bar{\F}_i\}} -Q_{ji}^\T \mb{1}_{j\not\in\F_i}\\
    &\ -L_{ji}^\T B_{jj}^\T - B_{ii}L_{ij}.
\end{aligned}
\end{equation}
The local state feedback controller gains (which includes the new feedback interconnections $\{K_{ij}:j\not\in \E_i\cap\N_{i-1}\}$ and $\{K_{ji}:j\not\in\F_i\cap\N_{i-1}\}$) are:
\begin{equation}\label{Eq:Th:StabilizationUnderFSF3}
    K_{ij} = L_{ij}M_{jj}^{-1} \ \ \mbox{ and } \ \ K_{ji} = L_{ji}M_{ii}^{-1}, \ \ \forall j\in\N_{i}
\end{equation} 
and the new system interconnections are: $\{A_{ij} : j\not\in\E_i\cap \N_{i-1}\}$ and 
$\{A_{ji}\triangleq Q_{ji}M_{ii}^{-1}:j\not\in\F_i\cap\N_{i-1}\}$.
\end{theorem}

\begin{proof}
    Let us define $M \triangleq \diag(M_{ii}:i\in\N_N)$, $W \triangleq -MA^\T-AM-L^\T B^\T-BL$ and $K=LM^{-1}$, where now $A=[A_{ij}]_{i,j\in\N_N}$ and $K=[K_{ij}]_{i,j\in\N_N}$ includes variable interconnection blocks $\{A_{ij}:j\not\in\E_i\}$, $\{A_{ji}:j\not\in\F_i\}$ and $\{K_{ij}:j\not\in\E_i\}$, $\{K_{ji}:j\not\in\F_i\}$, respectively (replacing the fixed $\0$ blocks that were there in $A$ \eqref{Eq:NetworkedSystemDynamics} and $K$ \eqref{Eq:GlobalControl}). 

    Starting from Lm. \ref{Lm:Stability}, it is easy to show that if there exists $M>0$ and $L$ such that $W>0$, the feedback controller gains given by $K=LM^{-1}$ guarantees the closed-loop stability of the networked system (with new interconnections). Based on the above definition of $W=[W_{ij}]_{i,j\in\N_N}$, it is a symmetric network matrix (see Def. \ref{Def:NetworkMatrices}) where  
    \begin{equation}\label{Eq:Th:StabilizationUnderFSFStep1}
        W_{ij} = -M_{ii}A_{ji}^\T - A_{ij}M_{jj} -L_{ji}^\T B_{jj}^\T - B_{ii}L_{ij}.
    \end{equation}
    Therefore, we can use Alg. \ref{Alg:DistributedPositiveDefiniteness} to enforce $W>0$ in a decentralized and compositional manner via enforcing $\tilde{W}_{ii}>0$ at each subsystem $\Sigma_i,i\in\N_N$ sequentially.

    As in the proof of Th. \ref{Th:CTLTIStability}, to make this local enforcement $\tilde{W}_{ii}>0$ an LMI problem, we need to replace any bi-linear term in $W_{ij}$ \eqref{Eq:Th:StabilizationUnderFSFStep1} using a change of variables. In this case, the term $-M_{ii}A_{ji}^\T$ in $W_{ij}$ \eqref{Eq:Th:StabilizationUnderFSFStep1} is bi-linear, and thus we introduce:
    \begin{equation}
        Q_{ji}^\T \triangleq M_{ii}A_{ji}^\T \iff Q_{ji} \triangleq A_{ji}M_{ii}, \ \ \forall j\not\in\F_i\cap\N_{i-1},
    \end{equation}
    which transforms $W_{ij}$ in \eqref{Eq:Th:StabilizationUnderFSFStep1} into the form \eqref{Eq:Th:StabilizationUnderFSF2}. Consequently, testing $\tilde{W}_{ii}>0$ takes the form of an LMI problem \eqref{Eq:Th:CTLTIStability}. The remainder of the proof is similar to Th. \ref{Th:CTLTIStability}, and is therefore omitted.
\end{proof}

\subsection{Enforcing Dissipativity and Dissipativate-ability}

Next, we provide decentralized and compositional interconnection topology synthesis techniques to ensure dissipativity and dissipativate-ability. In particular, we consider the $(Q,S,R)$-dissipativity property introduced in Def. \ref{Def:QSRDissipativity}, and regarding the specification matrices $Q,S$ and $R$, we assume: (1) they are network matrices, (2) $Q$ is a block diagonal, (3) $-Q>0$ (see Rm. \ref{Rm:QSRDissipativityVariants}), and (3) $R=R^\T$.

\begin{theorem} \label{Th:CTLTIQSRDissipativity}(Dissipativity via Topology Synthesis)
The networked system \eqref{Eq:NetworkedSystemDynamics} (where $C,D$ are block diagonal) under $w(t)=\0$ is $(Q,S,R)$-dissipative from $u(t)$ to $y(t)$ if at each subsystem $\Sigma_i,i\in\N_N$, the LMI problem  
\begin{equation}\label{Eq:Th:CTLTIQSRDissipativity}
\begin{aligned}
    \mathbb{P}_3: \text{Find}& \ \ P_{ii}, \{G_{ij}:j\not\in\E_i\}, \{A_{ji}:j\not\in\F_i\},\\
    & \ \ \{H_{ij}:j\not\in\E_i\}, \{B_{ji}:j\not\in\F_i\}\\
    \text{such that}& \ \ P_{ii} > 0, \ \ \tilde{W}_{ii}>0,
\end{aligned}
\end{equation}
is feasible, where $\tilde{W}_{ii}$ is computed from Alg. \ref{Alg:DistributedPositiveDefiniteness} (Steps: 3-16) when analyzing $W = [W_{ij}]_{i,j\in\N_N}>0$ with  
\begin{equation}\label{Eq:Th:CTLTIQSRDissipativity2}
    W_{ij} = 
    \bm{
     W_{ij}^{11} &  W_{ij}^{12}  & C_{ii}^\T e_{ij} \\
    W_{ij}^{21}          & D_{ii}^\T S_{ij} + S_{ji}^\T D_{jj} + R_{ij}    & D_{ii}^\T e_{ij} \\
C_{jj}e_{ij} & D_{jj}e_{ij} & -Q_{ii}^{-1}e_{ij}
    }.
\end{equation}
where
\begin{equation*}
    \begin{aligned}
        W_{ij}^{11} =&\ -P_{ii} A_{ij}\mb{1}_{\{j\in\bar{\E}_i\}} - G_{ij}\mb{1}_{\{j\not\in\E_i\}} - A_{ji}^\T P_{jj},\\
        W_{ij}^{12} =&\ -P_{ii}B_{ij}\mb{1}_{\{j\in\bar{\E}_i\}} - H_{ij}\mb{1}_{\{j\not\in\E_i\}} + C_{ii}^\T S_{ij},\\
        W_{ij}^{21} =&\ -B_{ji}^\T P_{jj} + S_{ji}^\T C_{jj}. 
    \end{aligned}
\end{equation*}
The new system interconnections are:
$\{A_{ij}\triangleq P_{ii}^{-1} G_{ij}: j\not\in\E_i\cap\N_{i-1}\}$ and 
$\{A_{ji}:j\not\in\F_i\cap\N_{i-1} \}$, and the new input interconnections are:
$\{B_{ij}\triangleq P_{ii}^{-1} H_{ij}: j\not\in\E_i\cap\N_{i-1}\}$ and 
$\{B_{ji}:j\not\in\F_i\cap\N_{i-1} \}$.
\end{theorem}
\begin{proof}
The proof starts by defining $P \triangleq \diag(P_{ii}:i\in\N_N)$ and 
\begin{equation*}
    W \triangleq BEW \big(
    \scriptsize\bm{
    -A^\T P - P A     &  -PB + C^\T S   & C^\T \\
    -B^\T P + S^\T C  & D^\T S + S^\T D + R    & D^\T \\ 
    C & D & -Q^{-1}    
    }\big) 
\end{equation*}
(inspired by \eqref{Eq:Lm:QSRDissipativity} and including the interested variable interconnection blocks), and proceeds using Prop. \ref{Lm:QSRDissipativity} in a similar manner to the proof of Th. \ref{Th:CTLTIStability} (note that here Lm. \ref{Lm:AlternativeLMI_BEW} is needed to deduce $W=\text{BEW}(\Psi)>0 \iff \Psi>0$). Therefore, explicit details of the proof are omitted here. 
\end{proof}

\begin{theorem}\label{Th:DissipativationUnderFSF}
(Dissipativate-ablity via Topology Synthesis) 
The networked system \eqref{Eq:NetworkedSystemDynamics} (where $B,C,F$ are block diagonal) under $D=\0$ and local state feedback control \eqref{Eq:LocalStateFeedbackControl} is $(Q,S,R)$-dissipative from $w(t)$ to $y(t)$ if at each subsystem $\Sigma_i,i\in\N_N$, the LMI problem    
\begin{equation}\label{Eq:Th:DissipativationUnderFSF}
\begin{aligned}
    \mathbb{P}_4: \text{Find}&\ \ M_{ii}, L_{ii}, \{L_{ij}:j\in\N_{i-1}\}, \{L_{ji}:j\in\N_{i-1}\},\\
    &\ \ \{A_{ij}:j\not\in\E_i\cap\N_{i-1}\},\{G_{ji}:j\not\in\F_i\cap\N_{i-1}\},\\
    &\ \ \{E_{ij}:j\not\in\E_i\cap\N_{i-1}\},\{E_{ji}:j\not\in\F_i\cap\N_{i-1}\},\\
    \text{such that}& \ \ M_{ii} > 0, \ \ \tilde{W}_{ii}>0, 
\end{aligned}
\end{equation}
is feasible, where $\tilde{W}_{ii}$ is computed from Alg. \ref{Alg:DistributedPositiveDefiniteness} (Steps: 3-16) when enforcing $W = [W_{ij}]_{i,j\in\N_N}>0$ with $W_{ij}$ given in \eqref{Eq:Th:DissipativationUnderFSF2}. 
The local state feedback controller gains (which include new feedback interconnections $\{K_{ij}:j\not\in\E_i\cap\N_{i-1}\}$ and $\{K_{ji}:j\not\in\F_i\cap\N_{i-1}\}$) are given by \eqref{Eq:Th:StabilizationUnderFSF3}, the new system interconnections are: 
$\{A_{ij}:j\not\in\E_i\cap\N_{i-1}\}$ and $\{A_{ji}\triangleq G_{ji}M_{ii}^{-1}:j\not\in\F_i\cap\N_{i-1}\}$, and the new input interconnections are:
$\{E_{ij}:j\not\in\E_i\cap\N_{i-1}\}$ and $\{E_{ji}:j\not\in\F_i\cap\N_{i-1}\}$.
\end{theorem}

\begin{proof}
The proof is similar to that of Th. \ref{Th:CTLTIQSRDissipativity}.
\end{proof}

\begin{remark}
    (Designing Intrinsic Parameters) When using the decentralized LMI problems proposed in Th. \ref{Th:CTLTIQSRDissipativity} and Th. \ref{Th:DissipativationUnderFSF}, if the application allows, we can also treat some subsystem intrinsic parameters (in addition to interconnection parameters) as design variables - without compromising the LMI format of the problem. For example, in the LMI problem \eqref{Eq:Th:CTLTIQSRDissipativity}, we can include $C_{ii}$ and/or $D_{ii}$ as design variables as $W_{ij}$ \eqref{Eq:Th:CTLTIQSRDissipativity2} is linear in both $C_{ii}$ and $D_{ii}$. 
\end{remark}

\begin{figure*}[!b]
    \centering
    \hrulefill
\begin{equation}\label{Eq:Th:DissipativationUnderFSF2}
    % \scriptsize
    W_{ij} = \bm{
    -A_{ij}M_{jj}-M_{ii}A_{ji}^\T\mb{1}_{\{j\in\bar{\F}_i\}}-G_{ji}^\T\mb{1}_{\{j\not\in\F_i\}}-B_{ii}L_{ij}-L_{ji}^\T B_{jj} & -E_{ij}+M_{ii}C_{ii}^\T S_{ij} & M_{ii}C^\T_{ii}e_{ij} \\ 
    -E_{ji}^\T+S_{ji}^\T C_{jj}M_{jj} & F_{ii}^\T S_{ij} + S_{ji}F_{jj} +R_{ij} & F_{ii}^\T e_{ij} \\
    C_{jj}M_{jj}e_{ij} & F_{jj} e_{ij} & -Q_{ii}^{-1}e_{ij} 
    }
    \end{equation}
\end{figure*}

%%%% Revised upto this

% \vspace{-1mm}
\section{Case Study}\label{Sec:CaseStudy}

In this section, we compare: (1) the decentralization-based topology synthesis (\textbf{DeTS}) approach proposed for linear networked systems in this paper with (2) the dissipativity-based topology synthesis (\textbf{DiTS}) approach proposed for non-linear networked systems in our recent work \cite{WelikalaP52022}. In particular, we use two variants of the \textbf{DiTS} approach based on the accuracy of the used dissipativity information of the subsystems. They are denoted by the acronyms \textbf{W-DiTS} and \textbf{S-DiTS}, representing scenarios where the available dissipativity information of the subsystems is weak (less precise) and strong (more precise), respectively. Note that, due to space constraints and for simplicity, we limit this case study to scenarios where topology synthesis is required to ensure the stability of a certain randomly generated  networked system.

\subsection{Considered Networked System}

In this case study, we consider a randomly generated networked system of the form \eqref{Eq:NetworkedSystemDynamics} with $N=5, n_i=3,\forall i\in\N_N$ and $B=E=\0$. The initial values of a subset of system matrices $\{A_{ij}:j\in\bar{\E}_i, i\in\N_N\}$ are given in \eqref{Eq:InitialSystemMatrices}. The corresponding initial interconnection topology is shown in the graph in Fig. \ref{Fig:NumericalExample} (left). It is worth noting that this initial networked system is unstable and the decentralized stability analysis proposed in \cite{WelikalaP32022} returns inconclusive. 

\begin{figure*}[!b]
\centering
    \hrulefill
    \begin{equation}\label{Eq:InitialSystemMatrices}
    \begin{aligned}
    &\scriptsize A_{22} = 
        \begin{bmatrix}
         -2.06 &   1.43 &  -1.83 \cr
         -1.79 &  -2.01 &   1.78 \cr
          1.50 &  -2.08 &  -2.00
       \end{bmatrix},\ 
       A_{23}
       \begin{bmatrix}
         -1.35 &  -1.62 &  -8.97 \cr
          1.51 &  -1.40 & -13.70 \cr
          8.99 &  13.68 &  -1.43
       \end{bmatrix},\ 
       A_{24} = 
       \begin{bmatrix}
         -2.96 &   0.12 &   2.47 \cr
          1.94 &  -4.44 &   0.11 \cr
          1.53 &   1.94 &  -2.98
       \end{bmatrix},\\ 
       &\scriptsize A_{55} = 
       \begin{bmatrix}
         -2.86 &  -0.56 &  -1.64 \cr
         -0.56 &  -2.20 &  -1.06 \cr
         -1.64 &  -1.06 &  -4.93
       \end{bmatrix},\ 
       A_{52} = 
       \begin{bmatrix}
         -2.54 &   1.13 &  -0.13 \cr
         -0.02 &  -1.56 &   1.40 \cr
          1.14 &   0.83 &  -1.78
       \end{bmatrix},\ 
       A_{54} = 
       \begin{bmatrix}
         -1.42 &   0.09 &  -0.61 \cr
          0.09 &  -0.30 &  -0.18 \cr
         -0.61 &  -0.18 &  -0.86
       \end{bmatrix}.
       % A_{45} =
       % \begin{bmatrix}
       %   -0.23 &  -0.06 &  -0.24 \cr
       %   -0.06 &  -0.42 &  -0.09 \cr
       %   -0.24 &  -0.09 &  -0.54
       % \end{bmatrix}.
       \normalsize
   \end{aligned}
    \end{equation}
    \normalsize
\end{figure*}

\begin{figure}[!h]
    \centering
    \includegraphics[width=3in]{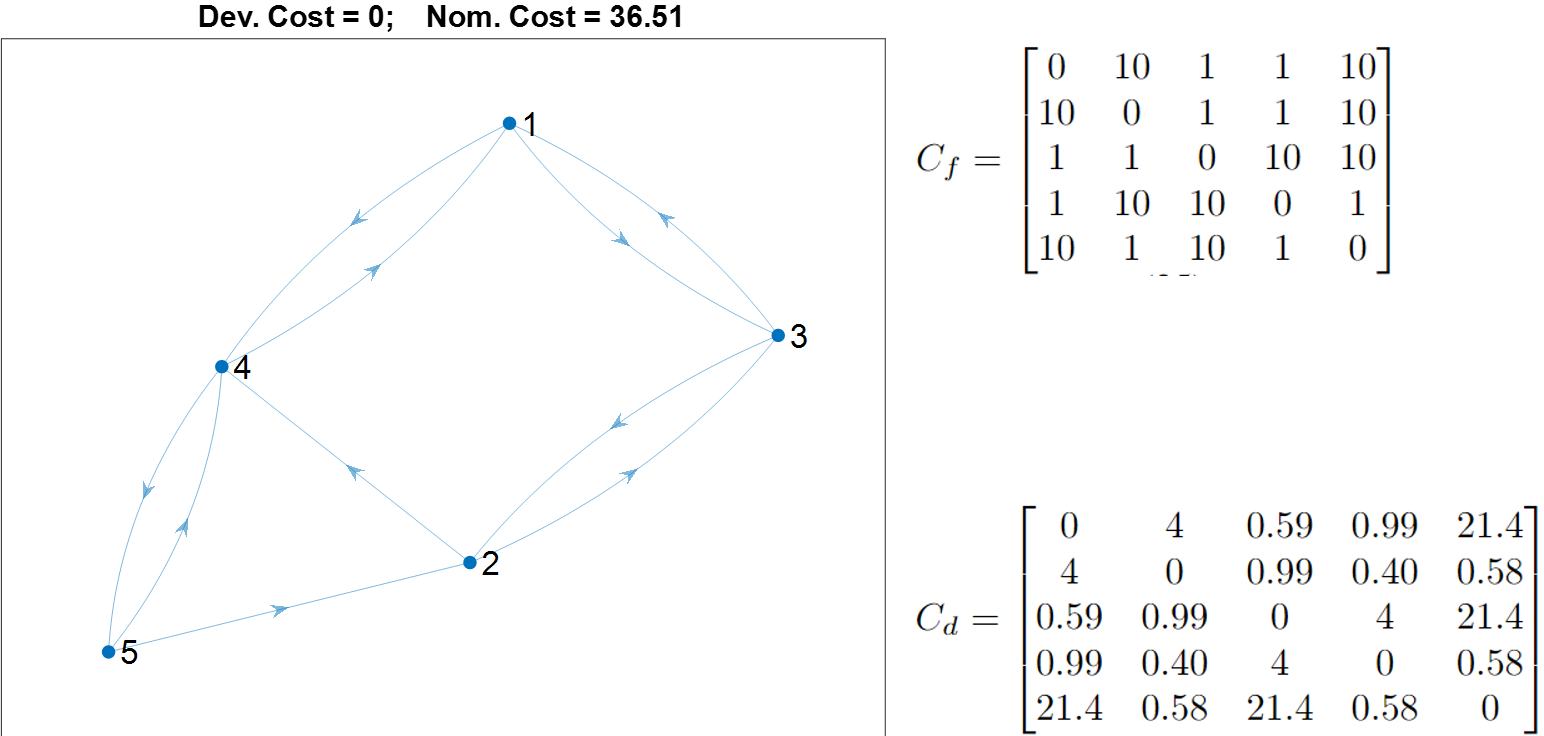}
    \caption{Numerical Example : Initial Topology}
    \label{Fig:NumericalExample}
\end{figure}

Two interconnection cost matrices inspired by this initial interconnection topology are shown in the same figure (i.e., $C_f$ and $C_d$). Note that $C_f$ has fixed cost levels while $C_d$ has graphical distance-inspired cost levels for different interconnections. Elements of these cost matrices are used in the topology synthesis processes (e.g., as $[c_{ij}]_{i,j\in\N_N}$ in \eqref{Eq:LMIObjective}) to penalize deviations from the initial topology. 

Since the proposed topology synthesis approach in this paper (\textbf{DeTS}) is decentralized, for comparison purposes, we use the following two centralized cost functions: 
\begin{equation}\label{Eq:CentralizedCostFunctions}
\begin{aligned}
    J_{Dev} \triangleq& \sum_{i,j\in\N_N, i\neq j} c_{ij} \Vert A_{ij}^*-\bar{A}_{ij}\Vert, \\ 
    J_{Nom} \triangleq& \sum_{i,j\in\N_N, i\neq j} \Vert A_{ij}^* \Vert, 
\end{aligned}
\end{equation}
where $c_{ij}$ is the interconnection cost coefficient (taken from either $C_f$ or $C_d$), $\bar{A}_{ij}$ is the initial (given) system matrix block, and $A_{ij}^*$ is the optimally synthesized system matrix block - corresponding to the interconnection $(\Sigma_j,\Sigma_i)$. Note that, in \eqref{Eq:CentralizedCostFunctions}, $J_{Dev}$ represents the weighted deviation from the initial topology while $J_{Nom}$ represents the nominal coupling strength of the synthesized topology.

% \begin{equation}
% \begin{aligned}
% C_d = \bm{
%     0 & 4 & 0.59 &  0.99  &  21.4\\
%     4 & 0 & 0.99 &  0.40  &  0.58\\
%     0.59 & 0.99 &   0     &    4  &  21.4\\
%     0.99 & 0.40  &  4     &    0  &   0.58\\
%     21.4 & 0.58  &  21.4  &   0.58 &   0}
% C_f = \bm{
%     0 & 10 & 1 & 1 & 10 \\
%     10 & 0 & 1 & 1 & 10\\
%     1  & 1 & 0 & 10 & 10\\
%     1 & 10 & 10 & 0 & 1\\ 
%     10 & 1 & 10 & 1 & 0}
% \end{aligned}
% \end{equation}

\subsection{Dissipativity Based Topology Synthesis (DiTS)}
The work in \cite{WelikalaP52022} considers networked systems comprised of non-linear subsystems $\tilde{\Sigma}_i:u_i\rightarrow y_i, i\in\N_N$ interconnected through a static interconnection matrix $M$ (e.g., see Fig. \ref{Fig:NonLinearNetworkedSystem}). A key advantage of the solution proposed in \cite{WelikalaP52022} is that it only requires the knowledge of $(Q,S,R)$-dissipativity properties of the subsystems: $\{(Q_i,S_i,R_i):i\in\N_N\}$ (in lieu of exact dynamic models of the subsystems). Even though subsystem dissipativity information is limited, it provides a simple, robust, and energy-based representation for the subsystems. In particular, the work in  \cite{WelikalaP52022} uses such subsystem dissipativity information to formulate an LMI problem so as to synthesize the optimal interconnection matrix $M$ (and hence, the interconnection topology) for the non-linear networked system under some minor assumptions. To make this paper self-contained, we have summarized this dissipativity-based topology synthesis approach in the following proposition. 

\begin{proposition} \cite[Prop. 5]{WelikalaP52022} \label{Pr:StabilizationNonLinear}
Under $R_i<0, \forall i\in\N_N$, a stabilizing interconnection matrix $M$ for the non-linear networked system shown in Fig. \ref{Fig:NonLinearNetworkedSystem} can be found via solving the LMI problem (centralized): 
\begin{equation}\label{Eq:Pr:StabilizationNonLinear}
\begin{aligned}
    \mathbb{P}_5: \text{Find}&\ \ L, \{p_i\in\R: i\in\N_N\}\\
    \text{such that}& \ \ p_i > 0, \forall i\in\N_N, \\
    & \ \ \bm{\textbf{R}_p & L \\ L^\T & -(L^\T \textbf{X} + \textbf{X}^\T L + \textbf{Q}_p)}
\end{aligned}
\end{equation}
as $M\triangleq \textbf{R}_p^{-1}L$, where  
$\textbf{R}_p \triangleq \diag(p_iR_i\I:i\in\N_N)$, 
$\textbf{Q}_p \triangleq \diag(p_iQ_i\I:i\in\N_N)$, and 
$\textbf{X} \triangleq \diag(R_i^{-1}S_i:i\in\N_N)$.
\end{proposition}

\begin{figure}[!t]
    \centering
    \includegraphics[width=1.75in]{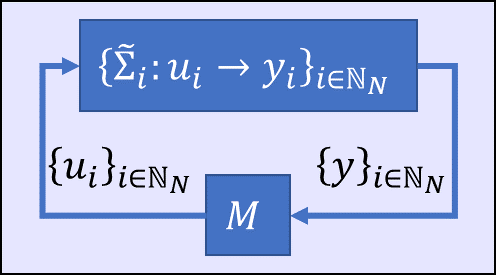}
    \caption{A non-linear networked dynamical system configuration considered in \cite{WelikalaP52022} (compare with Fig. \ref{Fig:LinearNetworkedSystem}).}
    \label{Fig:NonLinearNetworkedSystem}
\end{figure}
 
% How to apply this ... 
To apply Prop. \ref{Pr:StabilizationNonLinear} for the considered linear networked system, we first need to identify the corresponding construction of a nonlinear subsystem $\tilde{\Sigma}_i, i\in\N_N$ and the interconnection matrix $M=[M_{ij}]_{i,j\in\N_N}$ (shown in Fig. \ref{Fig:NonLinearNetworkedSystem}). This can be achieved by re-arranging the dynamics of a considered linear subsystem $\Sigma_i, i\in\N_N$ as:
\begin{equation*}
    \Sigma_i:\Big\{  \dot{x}_i(t) = \sum_{j\in\bar{\E}_i} A_{ij}x_j(t) = A_{ii}x_i(t) + \sum_{j\in\E_i} A_{ij} x_j(t).
\end{equation*}
Now, the dynamics of a corresponding non-linear (in name only) subsystem $\tilde{\Sigma}_i,i\in\N_N$ can be written as:
\begin{equation}\label{Eq:NonLinearSubsystem}
\tilde{\Sigma}_i:
\begin{cases}
        \dot{x}_i(t) =& A_{ii} x_i(t) + u_i(t),\\
        y_i(t) =& x_i(t),
\end{cases}
\end{equation}
where $u_i(t) \triangleq \sum_{j\in\E_i} A_{ij}y_j(t) \equiv \sum_{j\in\N_N} M_{ij} y_j(t)$. Therefore, $M_{ij} \triangleq A_{ij} \mb{1}_{\{i,j\in\N_N, i\neq j\}}$, and thus, using Prop. \ref{Pr:StabilizationNonLinear}, we can synthesize the system matrices $\{A_{ij}:i,j\in\N_N, i\neq j\}$ which implies the optimal interconnection topology for the considered networked system.

% Passivity properties
Recall that, per our As. \ref{As:FixedBlocks}, the system matrices $\{A_{ii},i\in\N_N\}$ and the non-linear subsystems $\{\tilde{\Sigma}_i:i\in\N_N\}$ \eqref{Eq:NonLinearSubsystem} are prespecified. However, to apply Prop. \ref{Pr:StabilizationNonLinear}, we only require the subsystem dissipativity properties: $\{(Q_i,S_i,R_i):i\in\N_N\}$. Here we assume each subsystem $\tilde{\Sigma}_i,i\in\N_N$ to have input and output passivity indices as $\nu_i$ and $\rho_i$, respectively. In other words, subsystem $\tilde{\Sigma}_i$ \eqref{Eq:NonLinearSubsystem} is assumed to be $(-\rho_i\I,0.5\I,-\nu_i\I)$-dissipative (see Def. \ref{Rm:QSRDissipativityVariants}). Candidate values for such passivity indices were estimated by applying Lm. \ref{Lm:QSRDissipativity} under: (1) a trial and error approach (which led to weak/less precise passivity indices), and (2) an LMI-based optimization approach \cite{WelikalaP52022} (which lead to strong/precise passivity indices). It is worth noting that there are several other alternative approaches to estimate such passivity indices (e.g., see \cite{WelikalaP42022,Koch2021,Tanemura2019}). Note that the said two types of passivity indices estimates led to the two dissipativity-based topology synthesis approaches: (1) \textbf{W-DiTS} and (2) \textbf{S-DiTS} mentioned earlier.

% Cost function
Note also that, inspired by Rm. \ref{Rm:LMIObjective}, to penalize deviations from the initial interconnection topology, when solving the LMI problem \eqref{Eq:Pr:StabilizationNonLinear} in Prop. \ref{Pr:StabilizationNonLinear}, we use the objective function
\begin{equation}
 J = \left\Vert [c_{ij} (L_{ij}-p_iR_i\bar{A}_{ij})]_{i,j\in\N_N} \right\Vert.   
\end{equation}

\subsection{Decentralization Based Topology Synthesis (DeTS)}
For the considered networked system in this case study, the application of the proposed \textbf{DeTS} approach is straightforward as it only involves solving the sequence of LMI problems given in Th. \ref{Th:CTLTIStability}. In the implementation, as the LMI objective function, we used \eqref{Eq:LMIObjective} (with $\beta_ii = \Vert P_jj\Vert$) as proposed in Rm. \ref{Rm:LMIObjective}. Note also that, as suggested in Rm. \ref{Rm:Refining}, we considered the possibility of refining all the existing interconnections. Consequently, in numerical results, we observed that it is possible (and even preferred) to get optimal interconnection topologies where some initial interconnections have been removed completely to preserve stability while minimizing deviations from the initial topology.

% Observed results under both types of dissipativity properties ... 
\subsection{Observations and Discussion}

Figure \ref{Fig:NumericalResults} illustrates the synthesized optimal interconnection topologies under the aforementioned topology synthesis techniques: \textbf{W-DiTS} (Figs. \ref{Fig:NumericalResults}ab), \textbf{S-DiTS} (Figs. \ref{Fig:NumericalResults}cd), and \textbf{DeTS} (Figs. \ref{Fig:NumericalResults}ef), and under the said interconnection cost matrices $C_f$ (left) and $C_d$ (right). Moreover, the observed deviation and nominal cost values proposed in \eqref{Eq:CentralizedCostFunctions} are given in the titles of the subfigures in Fig. \ref{Fig:NumericalResults}.

\begin{figure}[!h]
    \centering
    \begin{subfigure}[t]{0.23\textwidth}
        \centering
        \captionsetup{justification=centering}
        \includegraphics[width=1.6in]{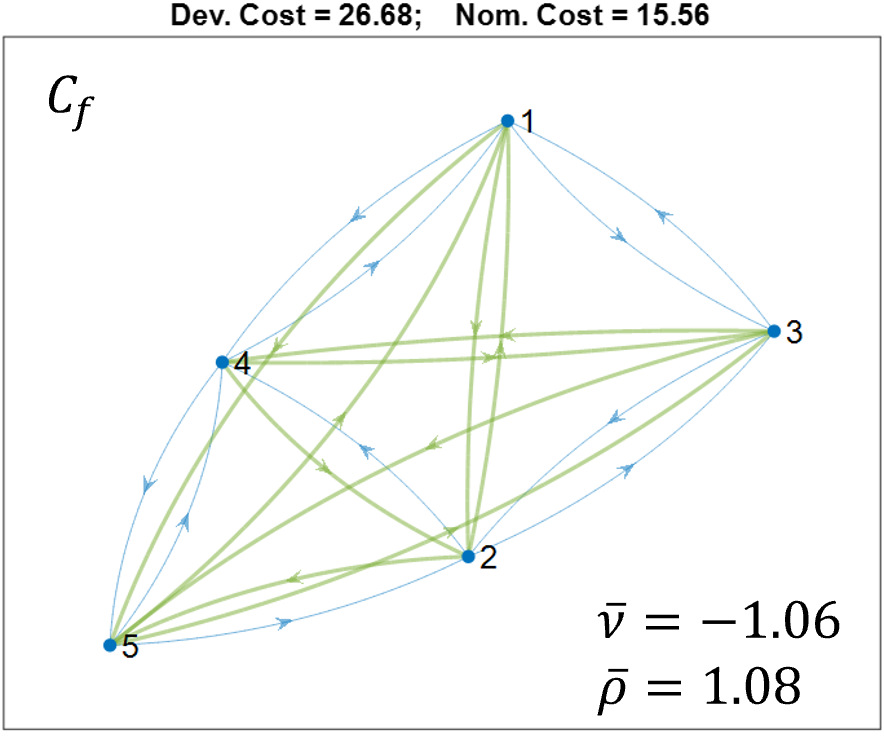}
        \caption{Under \textbf{W-DiTS} and $C_f$.}
        \label{Fig:NE_Fixed_NonLinear_WeakPassivity}
    \end{subfigure}
    \hfill
    \begin{subfigure}[t]{0.23\textwidth}
        \centering
        \captionsetup{justification=centering}
        \includegraphics[width=1.6in]{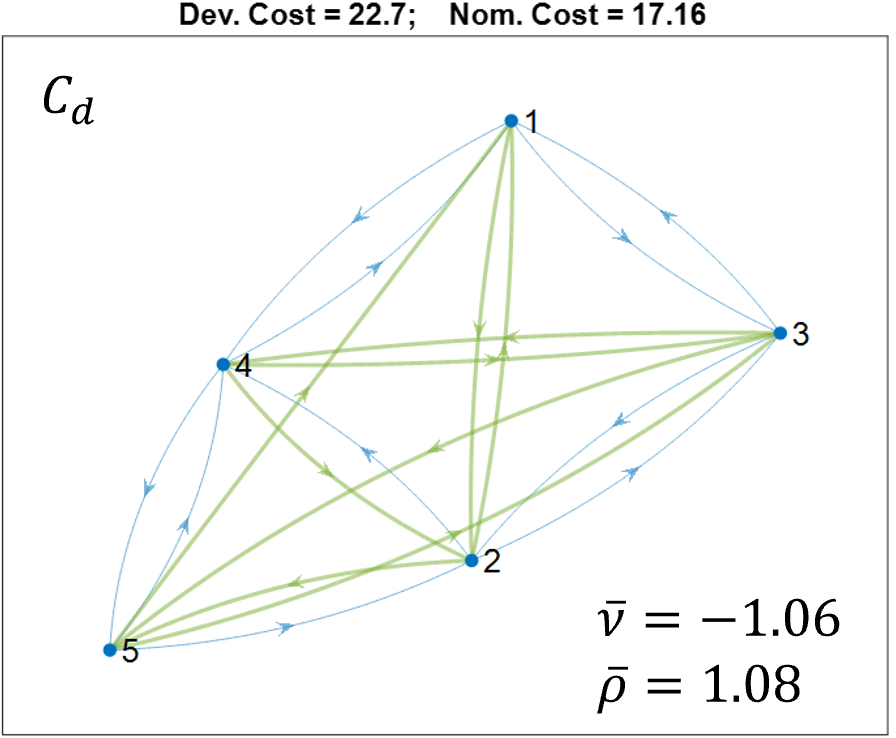}
        \caption{Under \textbf{W-DiTS} and $C_d$.}
        \label{Fig:NE_Dist_NonLinear_WeakPassivity}
    \end{subfigure}
    \begin{subfigure}[t]{0.23\textwidth}
        \centering
        \captionsetup{justification=centering}
        \includegraphics[width=1.6in]{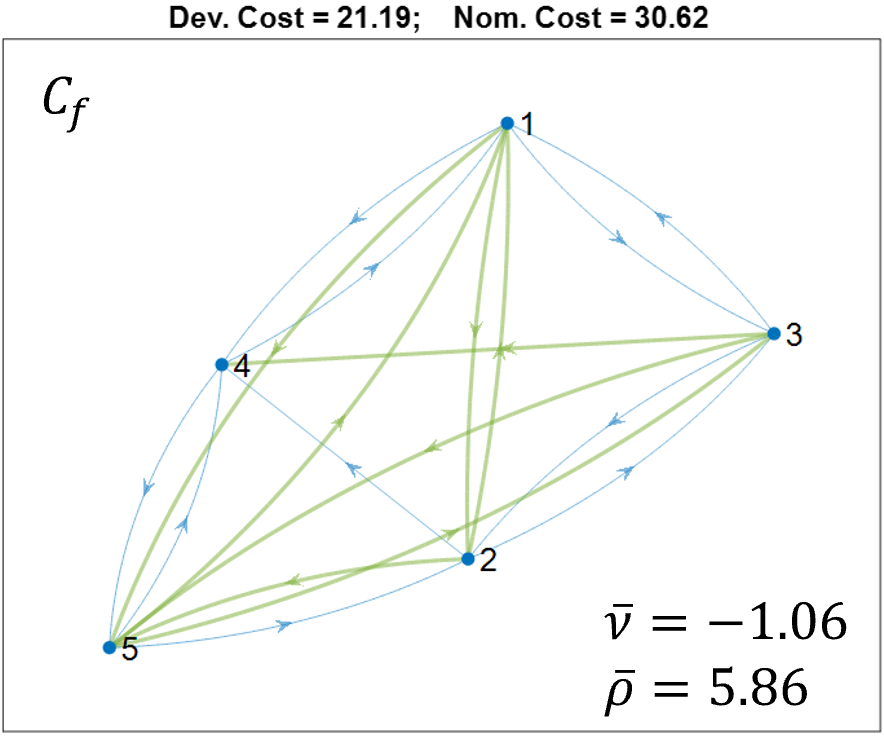}
        \caption{Under \textbf{S-DiTS} and $C_f$.}
        \label{Fig:NE_Fixed_NonLinear_StrongPassivity}
    \end{subfigure}
    \hfill
    \begin{subfigure}[t]{0.23\textwidth}
        \centering
        \captionsetup{justification=centering}
        \includegraphics[width=1.6in]{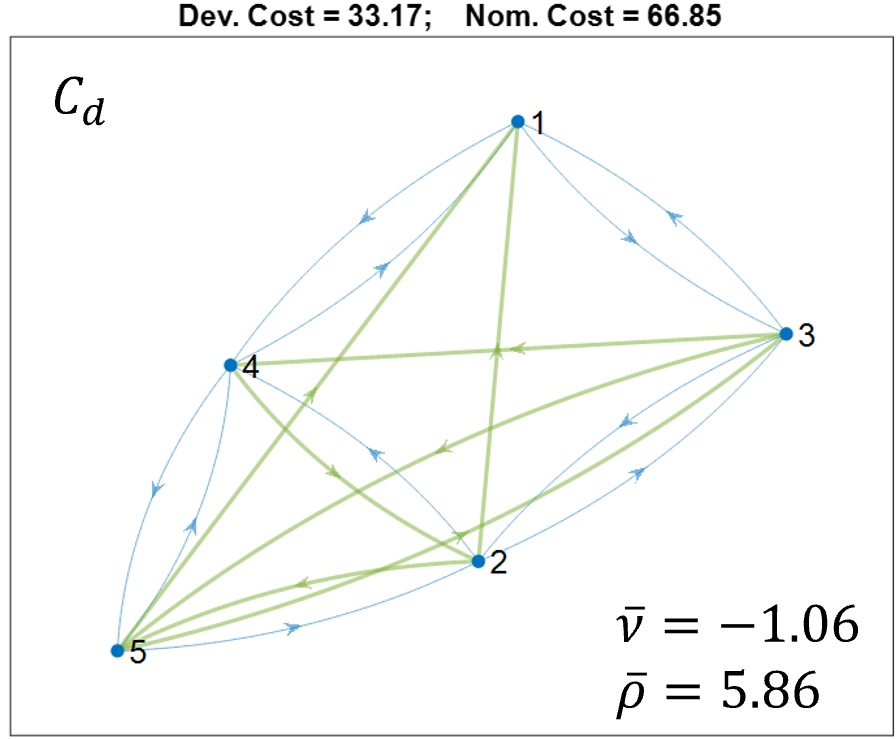}
        \caption{Under \textbf{S-DiTS} and $C_d$.}
        \label{Fig:NE_Dist_NonLinear_StrongPassivity}
    \end{subfigure}
    \begin{subfigure}[t]{0.23\textwidth}
        \centering
        \captionsetup{justification=centering}
        \includegraphics[width=1.6in]{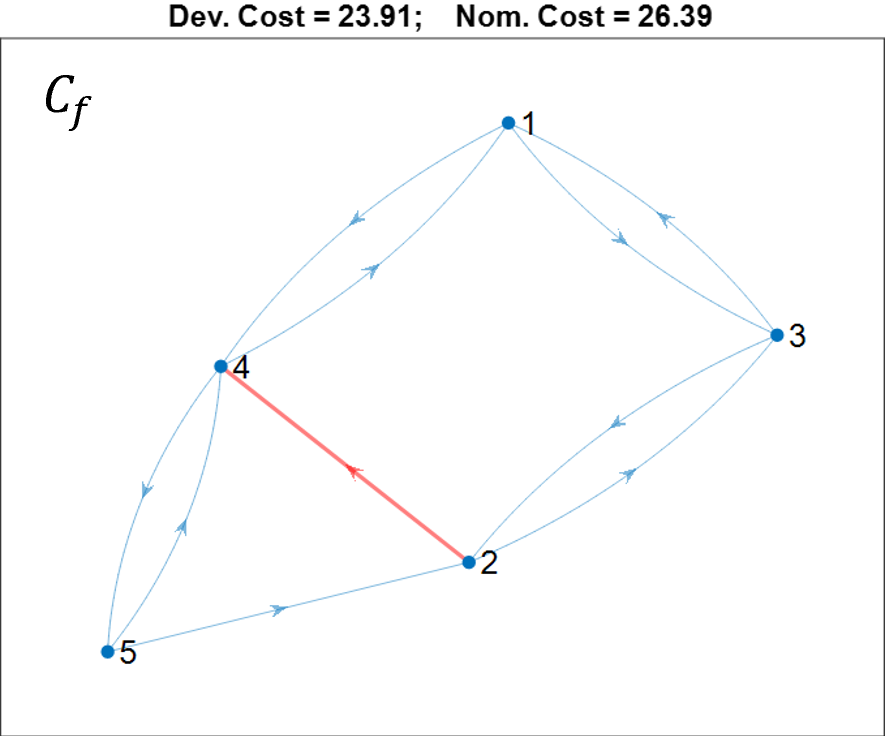}
        \caption{Under \textbf{DeTS} and $C_f$.}
        \label{Fig:NE_Fixed_Linear}
    \end{subfigure}
    \hfill
    \begin{subfigure}[t]{0.23\textwidth}
        \centering
        \captionsetup{justification=centering}
        \includegraphics[width=1.6in]{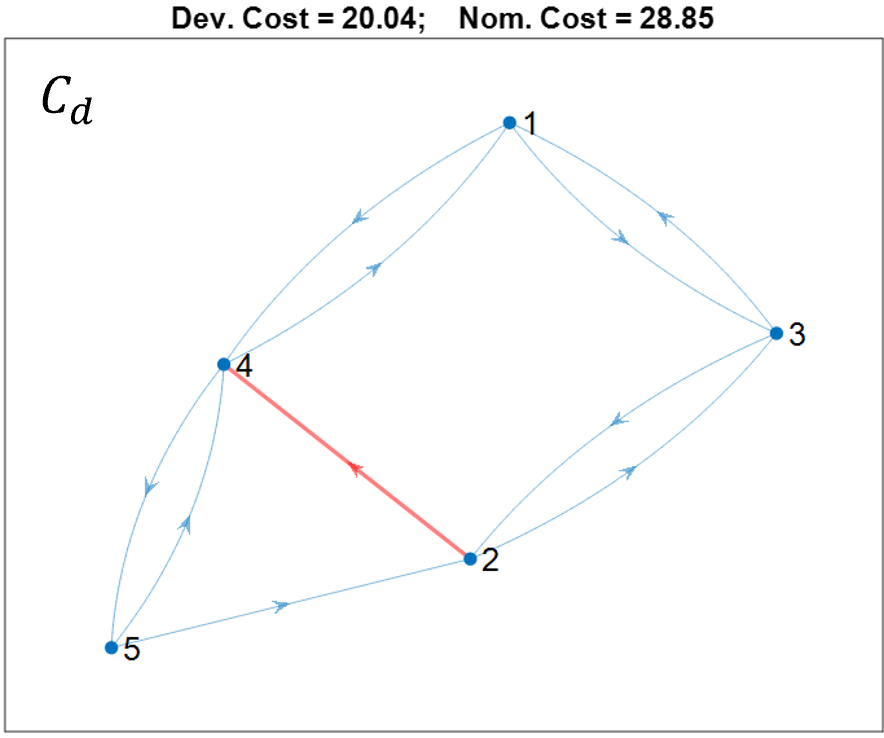}
        \caption{Under \textbf{DeTS} and $C_d$.}
        \label{Fig:NE_Dist_Linear}
    \end{subfigure}
    \caption{Obtained optimal interconnection topologies for the considered linear networked system under different topology synthesis methods: \textbf{W-DiTS}, \textbf{S-DiTS}, and \textbf{DeTS}, under different interconnection cost matrics: $C_f$ and $C_d$. The blue, green, and red edges in the sub-figures represent the initial, newly added, and entirely removed interconnections, respectively. The titles of the sub-figures indicate the deviation and nominal cost values  as defined in \eqref{Eq:CentralizedCostFunctions}}
    \label{Fig:NumericalResults}
\end{figure}  

Based on the observed newly added (green-colored) and entirely removed (red-colored) edges with respect to the initial topology (blue-colored) in each scenario, we can clearly see the practical advantage of the proposed \textbf{DeTS} approach in this paper compared to both \textbf{W-DiTS} and \textbf{S-DiTS} methods adapted from \cite{WelikalaP52022}. In essence, the proposed DeTS approach has mainly resorted to removing a single edge from the initial topology to ensure the stability of the considered networked system. In contrast, the dissipativity-based approaches \textbf{W-DiTS} and \textbf{S-DiTS} have mainly resorted to creating several new interconnections to achieve the same goal. Note also that both such approaches may have also refined some existing interconnections (this is also implied by the change in the nominal cost observed in Fig. \ref{Fig:NumericalResults}f compared to that in Fig. \ref{Fig:NumericalResults}e).

Another interesting observation is that when using the interconnection cost matrix $C_d$ as opposed to $C_f$, the number of newly added edges decreases (particularly the lengthy ones, e.g., compare Figs. \ref{Fig:NumericalResults}ac with Figs. \ref{Fig:NumericalResults}bd). However, during the same process, the observed nominal cost increases owing to the internal changes required to stabilize the considered networked system without using additional new edges.

Note also that a similar reduction in the number of newly added edges occurs when we have more precise/stronger passivity information regarding the subsystems (e.g., compare Figs. \ref{Fig:NumericalResults}ab with Figs. \ref{Fig:NumericalResults}cd). This is because dissipativity-based topology synthesis \cite{WelikalaP52022} becomes more flexible (as opposed to becoming more constrained/conservative) when underlying subsystems are strongly passive. 
This can also be understood by the fact that strongly passive systems not only can easily be stabilized but also can tolerate other connected weakly passive subsystems (due to the compositionality of passivity). Nevertheless, similar to before, an increment in the nominal cost can be seen due to the internal changes required to achieve stability without using additional new edges.  

In terms of the deviation cost \eqref{Eq:CentralizedCostFunctions}, when using $C_f$, the \textbf{S-DiTS} approach has provided the best performance. Note, however, that, in this case, the proposed \textbf{DeTS} approach performs closely to the \textbf{S-DiTS} approach while also having a better nominal cost. Moreover, when using $C_d$, the proposed \textbf{DeTS} approach performs the best (which is also the overall best deviation cost value observed in this case study). 

We conclude this paper by acknowledging some unique attributes of the \textbf{S-DiTS} approach adapted from \cite{WelikalaP52022} as opposed to the \textbf{DeTS} approach proposed in this paper. Let us consider the amount of information required to synthesize topologies under each approach. First, note that both these approaches use the initial interconnection system matrices $\{A_{ij}:i,j\in\N_N,i\neq j\}$ as a reference to penalize deviations from them. However, the intrinsic system matrices $\{A_{ii}\in\R^{n_i\times n_i}:i\in\N_N\}$ are only used in the proposed \textbf{DeTS} approach in this paper. In contrast, only two scalar passivity indices per each subsystem: $\{(\nu_i,\rho_i)\in\R^2: i\in\N_N\}$ are being used in the \textbf{S-DiTS} approach. In real-world scenarios, such scalar passivity indices can be estimated conveniently and accurately - compared to having to estimate the entire intrinsic system matrices. Moreover, as detailed in \cite{WelikalaP52022}, the \textbf{S-DITS} approach is applicable to a variety of networked systems comprised of non-linear subsystems.

\section{Conclusion}\label{Sec:Conclusion}
This paper considered networked systems comprised of interconnected linear subsystems and proposed a decentralized and compositional approach to stabilize or dissipativate such linear networked systems via synthesizing an optimal set of interconnections (topology) for the subsystems. The proposed topology synthesis approach was then extended to ensure the ability to stabilize or dissipativate linear networked systems using distributed feedback control. We gain this ability to optimally synthesize topologies by improving an existing decentralized and compositional approach for various analysis and controller synthesis tasks related to linear networked systems. The proposed topology synthesis approach only involves solving a sequence of linear matrix inequality problems - which can be implemented efficiently and scalably using standard convex optimization toolboxes. The presented case study showed that the proposed topology synthesis approach provides simplistic and high-performing solutions compared to an existing topology synthesis approach proposed for more general non-linear networked systems with limited information about the subsystems. Future work aims to closely study several critical real-world applications (e.g., supply chain networks, micro grids, vehicular platoons, and multi-robot systems) where the proposed topology synthesis approach can be directly applied to optimize existing networks in such applications.

% Apart from many generic linear networked systems applications (e.g., power grid control), a unique application for the proposed interconnection topology synthesis approach is in generating random stable (or dissipative, stabilizable, dissipativate-able) linear networked systems for simulation purposes.  
% We also include an interesting case study where the proposed interconnection topology synthesis approach is compared with an alternative approach that only uses dissipativity information of the involved subsystems. 

% In this paper, we considered several interesting networked system configurations comprised of interconnected sets of non-linear subsystems with known dissipativity properties. For these NSCs, centralized network analysis and interconnection topology synthesis techniques were developed as LMI problems. Since each proposed technique only uses subsystem dissipativity properties and takes the form of an LMI problem, they can be efficiently implemented and solved. Consequently, the proposed techniques in this paper are ideal for higher-level large-scale networked system design. Several supporting illustrative numerical results were also discussed. Future work aims to extend the proposed techniques for switched and hybrid networked systems. 

% \vspace{-1mm}
\bibliographystyle{IEEEtran}
\bibliography{References}

\end{document}